\documentclass[a4paper,UKenglish,cleveref, autoref, thm-restate]{lipics-v2021}
\title{Evaluating regular path queries under the all-shortest paths semantics} %TODO Please add

\author{Domagoj Vrgo\v{c}}{University of Zagreb, Croatia \and PUC Chile and IMFD, Chile}{domagoj.vrgoc@math.hr}{[orcid]}{ANID -- Millennium Science Initiative Program -- Code ICN17\_002}

\authorrunning{D. Vrgo\v{c}} %TODO mandatory. First: Use abbreviated first/middle names. Second (only in severe cases): Use first author plus 'et al.'

\Copyright{Domagoj Vrgo\v{c}} %TODO mandatory, please use full first names. LIPIcs license is "CC-BY";  http://creativecommons.org/licenses/by/3.0/

 %TODO mandatory; please add comma-separated list of keywords

\category{} %optional, e.g. invited paper

\relatedversion{https://arxiv.org/abs/2204.11137} %optional, e.g. full version hosted on arXiv, HAL, or other respository/website
\nolinenumbers %uncomment to disable line numbering

\usepackage{algorithm} 
\usepackage[noend]{algpseudocode}

\usepackage[textwidth=2cm,textsize=small]{todonotes}

\usepackage{framed}

\newcommand{\prob}{\textsc{Enumerate all shortest paths}}

\newcommand{\srpq}{\textsc{Single path RPQ eval}}
\newcommand{\arpq}{\textsc{All-shortest RPQ eval}}
\newcommand{\crpq}{\textsc{Count All-shortest}}

\usepackage{stmaryrd}

\newcommand{\semp}[2]{\llbracket #1\rrbracket_{#2}}

\usepackage{tikz}
\usetikzlibrary{automata, graphs,positioning,chains,arrows,decorations.pathmorphing}

\usepackage{balance}

\begin{document}

\maketitle

\begin{abstract}
The purpose of this paper is to explain how the textbook breadth-first search algorithm (BFS) can be modified in order to also create a compact representation of all shortest paths connecting a single source node to all the nodes reachable from it. From this representation, all these paths can also be efficiently enumerated. We then apply this algorithm to solve a similar problem in edge labelled graphs, where paths also have an additional restriction that their edge labels form a word belonging to a regular language. Namely, we solve the problem of evaluating regular path queries (RPQs) under the all-shortest paths semantics.
\end{abstract}

\section{Introduction}
\label{sec:intro}
Graph databases~\cite{AnglesABBFGLPPS18} have gained significant popularity in recent years, with multiple vendors offering graph database products~\cite{Webber12,TigerGraph,Stardog,AmazonNeptune,JenaTDB},   several efforts for defining the standard for graph data being under way~\cite{GQL,HarrisS13}, and the subject receiving substantial coverage in the research literature~\cite{MendelzonW89,Baeza13,HoganBC2020}. Some of the more popular graph formats in use today are edge-labelled graphs, which are the basis of the RDF standard~\cite{world2014rdf}, and property graphs~\cite{FrancisGGLLMPRS18}, which extend edge-labelled graphs by also allowing external annotations on nodes and edges. Graphs offer simple conceptualization of the application domain, modelling entities as nodes, and their relations as edges, with edge label indicating the connection type. Typical applications of graph database include crime detection network, transport networks, genealogy, etc. 

In terms of querying, graph databases offer some interesting challenges. Following \cite{AnglesABHRV17}, one could roughly divide graph queries into two main classes: (i) graph pattern; and (ii) path queries. Graph patterns look to match a smaller graph-shaped pattern into a graph database, while path queries allow exploring paths whose length is not known in advance, and thus require a limited form of recursion. A well studied class of path queries are regular path queries, or RPQs for short~\cite{Baeza13}. An RPQ specifies a regular language, and the query simply retrieves all pairs of nodes connected by a path whose edge labels form a word belonging to this language. RPQs are widely adopted in practice, forming part of the SPARQL standard for querying RDF data~\cite{HarrisS13}, and being supported in many popular graph engines~\cite{Webber12,TigerGraph,JenaTDB}. 

An interesting feature of RPQs is that their semantics, as usually defined in the research literature~\cite{Baeza13}, or in the SPARQL standard~\cite{HarrisS13}, only requires returning endpoints connected by a path satisfying a given regular expression, and not the path (or paths) witnessing this connection. However, in practice, returning paths can be a useful feature, which resulted in several engines supporting this option~\cite{Webber12,TigerGraph,Stardog,MillDB}, and the upcoming ISO standard for graph querying~\cite{GQL} specifying different options for returning paths. When returning paths, some of the available options would be to return: a single path, all paths, all simple paths, all shortest paths, all trails, etc. The traditional approach in research is to consider all paths (also termed all \emph{walks} in Graph Theory). However, in the presence of cycles, this can lead to an infinite amount of matching paths \cite{MendelzonW89}, which is undesirable in practice. For this reason, trails, or shortest paths, are favored by practical systems~\cite{FrancisGGLLMPRS18,Stardog,MillDB}.

In this paper, we will focus on solving the problem of evaluating RPQs in such a way that, in conjunction with all endpoint pairs connected by a path satisfying a given regular expression, we also return all shortest paths linking them. That is, we will evaluate RPQs under the all-shortest paths semantics. This feature is currently supported in some graph engines~\cite{FrancisGGLLMPRS18,Stardog,MillDB}, and is also being adopted by the graph standardization efforts~\cite{GQL}, so understanding how to solve this problem algorithmically seems rather relevant for practice.

\paragraph*{Contributions} In this paper we provide a simple algorithm for evaluating regular path queries under the all-shortest paths semantics, assuming that at least one of the query endpoints is known. We will write $q=(v,regex,?y)$ for such a query, signalling that the starting node of the path is fixed to $v$\footnote{Notice that this is done in order to anchor the search algorithm, and is adopted by some graph engines~\cite{Erling12}.}.  For this solution, we modify the standard breadth-first graph traversal, in order to be able to store all shortest paths as well. In essence, we are using breadth-first search in order to create a directed unlabelled graph which can compactly encode all of these paths, and allows them to be enumerated efficiently. This approach is often used in enumeration algorithms in databases and information extraction \cite{FlorenzanoRUVV20,GrezRU19}. We start by explaining how to obtain all shortest paths in an (unlabelled) graph, reachable from some fixed node. We then extend this algorithm in order to evaluate RPQs under the all-shortest paths semantics.

\paragraph*{Organization} The paper is structured as follows:
\begin{itemize}
\item In Section \ref{sec:prelim} we define graphs, graph databases, RPQs, and recall the BFS algorithm.
\item In Section \ref{sec:bfs}, we show how to extend the textbook BFS algorithm so that, starting from a fixed node $v$, we can return all $v'$ reachable from $v$, along with the list of all shortest paths connecting them.
\item In Section \ref{sec:algo}, we then apply this extension in order to enumerate all-shortest paths that witness an answer to an RPQ which has at least one endpoint fixed.
\item We conclude in Section \ref{sec:concl}.
\end{itemize}

\paragraph*{Proviso} The purpose of this paper is to solve a single database problem that occurs in practice. It is therefore short and to the point. It provides a theoretical algorithm for this problem and nothing else. In terms of novelty, we believe that the algorithm presented in Section 3 is folklore, and has been considered before. However, we could not locate a reference to the specific problem we are solving, so we include it here for completeness. To the best of our knowledge, using this algorithm to evaluate RPQs under the all-shortest paths semantics (Section \ref{ss:contrib}) is a new application not explored previously, at least not with its links to enumeration algorithms.

\section{Preliminaries}
\label{sec:prelim}
In this section we recall basic notions about graphs, graph databases, regular path queries. and the breadth-first graph traversals. We also make precise the notion of delay used in our enumeration algorithms.

\paragraph*{Graphs} A \emph{graph} $G$ is a pair $(V,E)$, where $V$ is a finite set of nodes, and $E\subseteq V\times V$ is a set of edges. A path in $G$ is simply a sequence $n_1\ldots n_k$ of nodes, such that $(n_i,n_{i+1})\in E$, for $i=1\ldots k-1$. 

\paragraph*{Graph databases} Following database theory literature \cite{AnglesABHRV17,Baeza13,MendelzonW89}, and the RDF specification~\cite{world2014rdf}, for a labelling alphabet $\Sigma$, we define edge-labelled graphs as follows: 

\begin{definition}
An \emph{edge-labelled graph} $G$ is a tuple $(V,E)$, where:
\begin{itemize}
\item $V$ is a finite set of nodes.
\item $E\subseteq V\times \Sigma \times V$ is a finite set of $\Sigma$-labelled edges.
\end{itemize}
\end{definition}

When talking about edge-labelled graphs, we will often refer to them as \emph{graph databases}. A \emph{path} in a graph database is a sequence $\pi = n_1 a_1 n_2 a_2 \ldots a_k n_{k+1}$, where $n_i\in V$, and $a_i\in \Sigma$, such that $(n_i,a_i,n_{i+1})\in E$, for $i=1\ldots k$. We define the label of a path $\pi$, denoted $\lambda(\pi)$ as the word $a_1\cdots a_k$, formed by concatenating the labels of the edges on this path.

\paragraph*{Querying graph databases} One of the most common class of queries over graph databases are regular path queries, or RPQs, which are well studied in the research literature \cite{Baeza13,MendelzonW89}, and are supported by the SPARQL standard~\cite{HarrisS13}, as well as many other graph query languages such as, for instance, Cypher~\cite{FrancisGGLLMPRS18} or GQL~\cite{GQL}.
For us, a \emph{regular path query} over a graph database $G$, will be an expression of the form $(v,regex,?x)$, where $v$ is a node in $G$, and $regex$ is a regular expression over $\Sigma$, the edge labelling alphabet of $G$\footnote{We will restrict the search to the single starting node. RPQs generally can also have a free variable in the place of $v$.}. Given an RPQ $q=(v,regex,?x)$ over a graph database $G$, its semantics is defined as the set of all nodes $v'$ in $G$, that are reachable from $v$ by a path $\pi$ such that $\lambda(\pi)$ is a word in the language of the regular expression $regex$. We denote the set of all answers for $q$ over $G$ with $\semp{q}{G}$.

While this version of RPQ semantics only asks for nodes reachable from $v$, one will often be also interested in retrieving the paths which witness this connection. In this regards, we will study two versions of evaluating RPQs: (i) the one which, for each node $v'\in \semp{q}{G}$, also returns a single shortest path witnessing this connection; and (ii) the one which, in addition to $v'$ returns all shortest paths linking $v$ and $v'$. We remark that the basic version of the semantics, returning only the reachable nodes, is the W3C standard for SPARQL property paths~\cite{HarrisS13}. Version (i) above would be a reasonable addition to the SPARQL standard, while the version (ii) is supported by the upcoming GQL standard~\cite{GQL} and Cypher~\cite{FrancisGGLLMPRS18}.

\paragraph*{Breadth-first search (BFS)} Given a (unlabelled) graph $G$, and a node $v$ in $G$, all nodes reachable from $v$, and a shortest path connecting them, can be obtained by the textbook BFS algorithm~\cite{cormen2022introduction}. We include the code of this version of BFS in Algorithm \ref{alg:bfs}.

 \begin{algorithm}
  \caption{All nodes in $G$ reachable from $v$, plus a shortest path}% using BFS/DFS.}
 \label{alg:bfs}
 \begin{algorithmic}[1]
 \Function{Search}{$G,v$}
 	\State $Open.init()$ \Comment{Empty queue(BFS)}
 	\State $Visited.init()$ \Comment{Empty set of visited nodes.}
	 \State $start \gets (v,\bot)$
 	\State $Open.push(start)$
 	\State $Visited.push(start)$
 	\While{!Open.isEmpty()}
 		\State current=Open.pop() \Comment{$current = (n,prev)$}
 		\State $solutions.add(n)$ \Comment{This is one of the reachable nodes}
 		\State $ReconstructPath(current)$
 		\For{$n'$ s.t. $(n,n')\in E$} \Comment{Neighbours of $n$}
	 		\If{$!(n'\in Visited)$} \Comment{Ignore $prev$}
	 			\State $next = (n',n)$
 				\State Open.push(next)
 				\State Visited.push(next)
 			\EndIf
 		\EndFor 
 	\EndWhile
 \EndFunction
  \end{algorithmic}
 \end{algorithm}
 
 Basically, Algorithm \ref{alg:bfs}, keeps track of pairs $(n',n)$, where $n'$ is a newly found node on some shortest path from $v$, and $n$ is its immediate predecessor on this shortest path. When searching whether $n'\in Visited$ (line 12), we only compare equality by first component of the pair stored in $Visited$. The function $ReconstrucPath(current)$ (line 10), for $current=(n,prev)$, reconstructs the shortest path from $v$ to $n$ by following $prev$ pointers in $Visited$. Namely, it starts in $n$, and then searches $Visited$ with $prev$, and so on, until reaching $\bot$. This is well defined since for $(n,prev)$ to be in $Visited$, there also must exist a pair $(prev,prev')$ in $Visited$.

\paragraph*{Output linear delay} Since we will deal with queries that can potentially have a large number of outputs, we will use the paradigm of \emph{enumeration algorithms} \cite{bagan2006mso,Segoufin13,FlorenzanoRUVV20,GrezRU19} to measure the efficiency of our solutions. Enumeration algorithms work in two phases: first, a \emph{pre-processing phase} is permitted in order to construct some data structures. Following this, the \emph{enumeration phase} uses these data structures to enumerate the solutions, one by one, and without repetitions. The efficiency of an enumeration algorithm is measured by determining both the complexity of the preprocessing phase,  and the  delay  between any two elements produced during the enumeration phase. The ideal case, when the delay does not depend on the input, is called \emph{constant delay}. Unfortunately, when outputting paths, the length of the path itself can depend on the input (the graph), so constant delay seems unachievable. Instead, we will aim to have \emph{output-linear delay}, where the delay is linear in the size of each output element. For instance, when outputting paths, this means that the time needed to output a single path is linear in its length (i.e. the number of nodes in the path), and the delay before starting to write down the next path is constant. This, in a sense, is optimal, since we have to at least write down each element in the path. 

\section{Enumerating all shortest paths from a single source}
\label{sec:bfs}
In this section we assume that $G$ is an (unlabelled) graph. That is, $G=(V,E)$, with $E\subseteq V\times V$. The problem we will solve is as follows:

\begin{framed}
  \begin{tabular}{ll}
    \textsc{Problem}: & \prob\\
    \textsc{Input}: & A graph $G = (V,E)$, and a node $v\in V$.\\
    \textsc{Output}: & For each node $v'$ reachable from $v$ in $G$, the list of all shortest paths \\
    &  connecting them.
  \end{tabular}
\end{framed}

In order to solve \prob, we will modify the standard breadth-first search algorithm~\cite{cormen2022introduction}, which, by default, also returns a single shortest path for each reachable node. The code of algorithm for \prob\ is given in Algorithm \ref{alg:plain}. We believe this algorithm to be folklore, however, we could not find a reference so we include it for completeness. The main difference between Algorithm~\ref{alg:bfs} and Algorithm~\ref{alg:plain} will be in bookkeeping. Namely, where standard BFS only keeps tracks of nodes it visited thus far in the set $Visited$, plus a pointer to a node used to reach the visited node, this extended version will need to store triples of the form $(n,dist,prevList)$ in $Visited$. The intuition here is that:
\begin{itemize}
\item $n$ is a node we already reached by some path;
\item $dist$ is the length of the shortest path from $v$ to $n$; and
\item $prevList$ keeps track of all the nodes that are an immediate predecessor of $n$, on some path of length $dist$ starting in $v$.
\end{itemize}

 \begin{algorithm*}[tbh]
  \caption{\prob\ for $G,v$}
 \label{alg:plain}
 \begin{algorithmic}[1]
 \Function{Search}{$G,v$}
 	\State $Open.init()$ \Comment{Empty queue.}
 	\State $Visited.init()$ \Comment{Empty set}
	 \State $start \gets (v,0,\bot)$
 	\State $Open.push(start)$
 	\State $Visited.push(start)$
 	\While{!Open.isEmpty()}
 		\State current=Open.pop() \Comment{$current = (n,depth,prevList)$}
			\State $solutions.add(n)$ \Comment{All shortest paths from $v$ already reached $n$}
			\State $reconstructPaths(current)$ \Comment{Enumerate all shortest paths}

 		\For{$n'$ s.t. $(n,n')\in E$}
	 		\If{$!(n' \in$ Visited$)$} \Comment{prevList or depth are not compared for equality}
				%\State $next.pred\gets current$
				\State new = $(n',depth+1,prevList.begin = prevList.end = current)$
 				\State Open.push(new)
 				\State Visited.push(new)
 			\EndIf
 			\If{$n' \in$ Visited} 
 			    \State new = Visited.get($n'$) \Comment{$new = (n',depth',prevList')$} 
 			    \If{$depth' == depth+1$} \Comment{Another shortest path to $n'$}
 			        \State prevList'.end$->$next = current
 			        \State prevList'.end = current 
 			    \EndIf
 			\EndIf
 		\EndFor 
 	\EndWhile
 \EndFunction
  \end{algorithmic}
 \end{algorithm*}

Basically, where standard BFS would only keep a single pointer to a previous node on the shortest path, our version will need to store a list of previous nodes on any shortest path. For our purposes, we will assume that $prevList$ is a linked list of pointers to triples of the form $(n,dist,prevList)$ that are stored in the set $Visited$. We assume that $prevList$ has methods $prevList.begin$, retrieving the first element, $prevList.end$, retrieving the final element, and each element in $prevList$ has a $next$ pointer to the following element in the list. For the set $Visited$ we assume that $n$ is the search key. Namely, we write $n\in Vistied$ to denote that there is a triple $(n,depth,prevList)$ in $Visited$. In order to retrieve this element we will use $Visited.get(n)$.

The algorithm works as standard BFS (e.g. as in Algorithm \ref{alg:bfs}) when a new node is reached for the first time (lines 12 -- 15). However, when the same node is revisited (line 16), if we reach it with a path that is of the same length as the shortest path to this node, we will simply add the predecessor we are expanding to the $prevList$ of the reached node (lines 19 -- 20). In essence, our algorithm creates a directed acyclic graph (encoded in $prevList$s of elements in $Visited$), which allows enumerating all shortest paths to a reached node. Upon popping a node from the queue (line 8), we know that all shortest paths already reached this node (in previous iteration we exhaust all the path of length needed to reach the node), so we can add it as a result, and print out all the shortest paths reaching it. This is done using the $prevList$, which in turn has pointers to other $prevList$s of predecessor nodes on some shortest path, etc. The shortest paths can then be recovered by doing a depth-first traversal of the constructed DAG, with all path ending in the node $(v,0,\bot)$, which marks the beginning of the path. We illustrate how the algorithm works via the following example.

\begin{example} Consider the graph of Figure \ref{fig:graph} (left).  Running Algorithm \ref{alg:plain} over this graph will result in the  DAG displayed in Figure \ref{fig:graph} (right), being stored in $Visited$. The DAG is constructed dynamically, and used to enumerate all shortest paths between $v$ and any node reachable from $v$. A pair $(n,dist)$ signals that the node $n$ was reached by a shortest path of length $dist$. The $prevList$ is illustrated in Figure \ref{fig:graph} (right) as arrows pointing to previous elements stored in $Visited$. That is, the element stored in $Visited$ consists of $(n,dist)$, and the list of pointers to previous elements. For instance, node $n_4$ can be reached from $v$ by three different paths of length 2. This is illustrated by having three backwards pointers to $(n_1,1), (n_2,1)$, and $(n_3,1)$. 
Notice that the enumeration in Algorithm \ref{alg:plain} is done upon popping the node from the queue (line 10), since at this point we already explored all possible shortest paths reaching it. In a sense, the enumeration is done before finishing the construction of the entire DAG containing all the answers, since one keeps on adding new elements for longer paths. \qed	 
\end{example}

\begin{figure*}[b]
\vspace*{0.3cm}
\resizebox{\textwidth}{!}{
\begin{tikzpicture}[->,>=stealth',auto, thick, scale = 1.0,initial text= {},    a/.style={ circle,inner sep=2pt    }]
 
		  % The graph
		  \node [state] at (0,0) (q0) {$v$};
		  \node [state] at (2.5,1.5) (q1) {$n_1$};
		  \node [state] at (2.5,0) (q2) {$n_2$};
		  \node [state] at (2.5,-1.5) (q3) {$n_3$};	  
		  \node [state] at (5,0) (q4) {$n_4$};
		  \node [state] at (7.5,0) (q5) {$n_5$};
  
		  % Graph edges
		  \path[->] (q0) edge  (q1);
		  \path[->] (q0) edge  (q2);
		  \path[->] (q0) edge  (q3);
		  \path[->] (q1) edge  (q4);
 		  \path[->] (q2) edge  (q4);
		  \path[->] (q3) edge  (q4);
		  \path[->] (q4) edge  (q5);
		  
% The DAG
        \begin{scope}[xshift=10cm]
        
        	\node [a] at (0,0) (b) {$\bot$};
        	\node [a] at (2,0) (v) {$(v,0)$};
        	\node [a] at (4,1.5) (n1) {$(n_1,1)$};
        	\node [a] at (4,0) (n2) {$(n_2,1)$};
        	\node [a] at (4,-1.5) (n3) {$(n_3,1)$};
        	\node [a] at (6,0) (n4) {$(n_4,2)$};
        	\node [a] at (8,0) (n5) {$(n_5,3)$};
        	
		  \path[->] (v) edge  (b);
		  \path[->] (n1) edge  (v);		  
		  \path[->] (n2) edge  (v);
		  \path[->] (n3) edge  (v);
		  \path[->] (n4) edge  (n1);		  
		  \path[->] (n4) edge  (n2);		  
		  \path[->] (n4) edge  (n3);		  		  		  
		  \path[->] (n5) edge  (n4);		  		  
		
		\end{scope}

		 \end{tikzpicture}
}
\caption{A graph (left), and the DAG constructed by Algorithm \ref{alg:plain} to enumerate all shortest paths from $v$ to any reachable node (right). In the DAG, the $prevList$ component is illustrated via arrows. For instance, the three arrows from $(n_4,2)$ correspond to the three pointers on its $prevList$.}
\label{fig:graph}
\end{figure*}
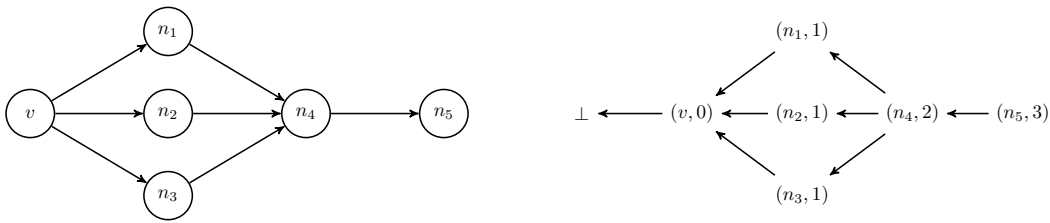

Analysing the running time of the algorithm, we obtain the following:

\begin{theorem}[Folklore]\label{theo:enum}
Given a graph $G$, and a node $v$, Algorithm \ref{alg:plain} solves \prob\ in time $O(|V|+|E|+|O|)$, where $O$ is the size of the output; that is, the number (and length) of all shortest paths starting in $v$.
\end{theorem}
\begin{proof}
In terms of complexity analysis, the algorithm runs similarly as standard BFS, covering the $O(|V|+|E|)$ portion. The main difference is in $ReconstructPaths$ function which, starting at $current = (n,depth,prevList)$ uses the linked list $prevList$ in order to enumerate all the shortest paths from $v$ to $n$. The number of such paths can be exponential, but a depth-first traversal of the DAG stored in $Visited$ and accessed through $prevList$ element of each triple in $Visited$ can be navigated in time proportional to the length and number of shortest paths between $v$ and $n$. Notice that here each path in this DAG contains such a shortest path, so the traversal can be done efficiently by remembering each position in $prevList$ from which a fork is possible. A detailed explanation of such ``smart" backtracking can be found in \cite{GrezRU19}.
\end{proof}

\paragraph*{Connection with enumeration algorithms}

It is interesting to notice that the logic of the algorithm closely follows the idea of enumeration algorithms used in databases and information extraction \cite{GrezRU19,FlorenzanoRUVV20}. In these algorithms one usually has a preprocessing phase that builds a DAG from which the answers can be enumerated, similarly to our approach in Algorithm \ref{alg:plain}. In these scenarios, the entire DAG is constructed, and only then the answers are enumerated. In our scenario, this would amount to starting the enumeration from each node of the created DAG. Of course, we can already enumerate all the shortest paths when a node is popped from the stack (lines 8 -- 10 in Algorithm \ref{alg:plain}). In the context of enumeration algorithms, this is known as early output, and it allows returning the answers as soon as they are discovered, which is a desirable feature in a classical database pipeline. Notice that a set of shortest paths can actually be returned earlier,  namely each time a node is (re)visited (line 12 and line 16), but in this case the paths will not be grouped according to the reached node. One difference between our algorithm and classical enumeration algorithms is also that in enumeration algorithms, the search is done in a graph in which only some reachable nodes are targets. This phenomenon also occurs when querying graph databases, which we explore in the next section.

\smallskip

Through the looking glass of enumeration algorithms~\cite{FlorenzanoRUVV20}, we can alternatively phrase Theorem~\ref{theo:enum} as follows:

\begin{corollary}
The problem \prob\ can be solved with output linear delay and $O(|V|+|E|)$ preprocessing.
\end{corollary}
\begin{proof}
As discussed above, this amounts to \emph{not} returning query answers as soon as they are available, but rather constructing the entire $Visited$ DAG using Algorithm~\ref{alg:plain}, which takes $O(|V|+|E|)$ time. Following this, we can simply enumerate the answers by starting in every node in the $Visited$ data structure. One could also keep a list of these nodes to start enumeration more easily. As explained previously, enumerating the shortest paths is done via the DFS algorithm, and takes time proportional to length of each path, which is in a sense optimal, since we must at the very least write down each output path symbol by symbol.
\end{proof}

\section{Adapting search algorithms for graph databases}
\label{sec:algo}
In this section we extend the BFS algorithm in order to answer regular path queries (RPQs) over graph databases. For this, we begin by covering the case when only a single shortest path is returned, and then extend the presented algorithm to return all shortest paths. Recall that, given a $\Sigma$-labelled graph $G$, we only consider regular path queries of the form:
$$q = (v,regex,?x)$$
where $v$ is a node of $G$, and $regex$ is a regular expression over $\Sigma$. Namely, we assume that the starting point of an RPQ is fixed. Should we wish to extend our results to cover the case of queries where the starting point is unknown, we can simply run the proposed algorithms from every node of the graph $G$.

\subsection{Single path version of RPQs}

We first recall how RPQs can be evaluated in the first version of the semantics. Namely, we will solve the following problem:

\begin{framed}
  \begin{tabular}{ll}
    \textsc{Problem}: & \srpq\\
    \textsc{Input}: & A graph database $G$, and an RPQ $q = (v,regex,?x)$ over $G$.\\
    \textsc{Output}: & For each node $v'\in \semp{q}{G}$, a single shortest path witnessing this. 
  \end{tabular}
\end{framed}

One method to evaluate RPQs, long standing in the theoretical literature \cite{MendelzonW89,Baeza13}, involves the cross-product construction, often known as the \emph{product graph}. Given a graph database $G = (V,E)$, and an RPQ $q = (v,regex,?x)$ over $G$, the product graph is constructed by first converting the regular expression $regex$ into an equivalent non-deterministic finite state automaton $(Q,\Sigma,\delta,q_0,q_F)$. Here $Q$ is a set of states, $\Sigma$ a finite alphabet of edge labels (same as i $G$), $\delta$  the transition relation over $Q\times \Sigma\times Q$, and $q_0,q_F$ the initial and final state, respectively. There are several standard ways of converting an expression to an equivalent automaton, for instance using Thompson's, or Ghluskov's construction \cite{sakarovitch2009elements}. In this paper, we will assume that the automaton has a single initial, and a single final state, and that no $\varepsilon$-transitions are present. The product graph $G_\times$, is then constructed as the graph database having $V_\times = V\times Q$ as the set of nodes. The edges, denoted $E_\times$,  are defined for each $(x,q_x)$ and $(y,q_y)$, and a label $l$, whenever: (i) $(x,l,y)\in E$;  and (ii) $(q_x,l,q_y)\in \delta$.  The answers to the query can now be found by running the BFS algorithm (Algorithm~\ref{alg:bfs}) on $G_\times$, but now taking $(v,q_0)$ as the initial node, and only returning answers when reaching a node paired with $q_F$, the accepting state of the NFA for $regex$.

 \begin{algorithm}[t]
  \caption{\srpq\ for a graph database $G$ and an RPQ $q= (v,regex,?x)$}% using BFS/DFS.}
 \label{alg:search}
 \begin{algorithmic}[1]
 \Function{Search}{$G,q$}
 	\State $\mathcal{A}\gets Automaton(regex)$
 	\State $Open.init()$ \Comment{Empty queue}
 	\State $Visited.init()$ \Comment{Empty set}
	 \State $start \gets (n,q_0,\bot)$
 	\State $Open.push(start)$
 	\State $Visited.push(start)$
 	\While{!Open.isEmpty()}
 		\State current=Open.pop() \Comment{$current = (n,q,prev)$}
 				\If{$q== q_F$} \Comment{A solution is found}
 					\State $solutions.add(n)$
 					\State $ReconstructPath(current)$
 				\EndIf
 		\For{$neighbour = (n',q') \in Neighbours(current)$}
	 		\If{$!neighbour \in Visited$}
				%\State $next.pred\gets current$
				\State $next = (n',q',n)$ 				
 				\State Open.push($next$)
 				\State Visited.push($next$)
 			\EndIf
 		\EndFor 
 	\EndWhile
 \EndFunction
  \end{algorithmic}
 \end{algorithm}
 
 For completeness, we include this algorithm in Algorithm \ref{alg:search}. Similar solutions have previously been proposed in the literature \cite{BaierDRV17,FiondaPG15}. As explained, this basically amounts to running Algorithm \ref{alg:bfs} on $G_\times$, taking $(v,q_0)$ as the start state of the search (line 5), and only returning a result if a node of the form $(n,q_F)$ is reached in $G_ \times$ (line 10). In Algorithm \ref{alg:search}, the set $Visited$ stores triples of the form $(n,q,prev)$; that is, the node $(n,q)\in G_\times$, and a pointer to the node used to reach this node on the shortest path from $(v,q_0)$. The search key for $Visited$ (used in line 14) is just the pair $(n,q)$. The one subtle detail is how the function $Neighbours(current)$ is computed (line 13). Basically, we start in $current=(n,q,prev)$, and simply retrieve all the neighbours of $(n',q')$ of the node $(n,q)$ in $G_\times$. That is, look for a label $a$, such that: (i) $(n,a,n') \in E$; and (ii) $(q,a,q')\in \delta$. Paths are reconstructed (line 12) in the same manner as for standard BFS. Notice that we assume $\mathcal{A}$ to have a single accepting state, which allows to avoid generating the same solution twice.
 
 \begin{algorithm*}[t]
  \caption{\arpq\ for a graph database $G$ and an RPQ $q= (v,regex,?x)$}% using BFS/DFS.}
 \label{alg:paths}
 \begin{algorithmic}[1]
 \Function{Search}{$G,q$}
 	\State $\mathcal{A}\gets Automaton(regex)$ \Comment{$q_0$ initial, $q_F$ final}
 	\State $Open.init()$ \Comment{Empty queue(BFS).}
 	\State $Visited.init()$ \Comment{Empty set of visited nodes.}
	 \State $start \gets (v,q_0,0,\bot)$
% 	\If{$q_0 == q_F$} $solutions.add(n)$ \Comment{Initial state is final.}
 	%\EndIf
 	\State $Open.push(start)$
 	\State $Visited.push(start)$
 	\While{!Open.isEmpty()}
 		\State current=Open.pop() \Comment{$current = (n,q,depth,prevList)$}
		\If{$q== q_F$} \Comment{We reached a solution}
			\State $solutions.add(n)$ \Comment{All shortest paths already reached $n$}%; multiplicities can be computed from prevList, or paths reconstructed by backtracking.}
			\State $reconstructPaths(current)$ \Comment{Reconstruct shortest paths to $n$}
		\EndIf

 		\For{next=$(n',q')$ $\in$ Neighbours(current)}
	 		\If{!(next) $\in$ Visited} \Comment{prevList or depth are not compared for equality}
				%\State $next.pred\gets current$
				\State new = $(n',q',depth+1,prevList.begin = prevList.end = current)$
 				\State Open.push(new)
 				\State Visited.push(new)
 			\EndIf
 			\If{next=(n',q') $\in$ Visited} 
 			    \State new = Visited.get(n',q') \Comment{new = (n',q',depth',prevList')} 
 			    \If{depth' == depth+1} \Comment{Another shortest path to (n',q')}
 			        \State prevList'.end$->$next = current
 			        \State prevList'.end = current \Comment{This updates the values in Visited}
 			    \EndIf
 			\EndIf
 		\EndFor 
 	\EndWhile
 \EndFunction
  \end{algorithmic}
 \end{algorithm*}

\subsection{Returning all shortest paths}
\label{ss:contrib}

Here we look at the following version of the problem:

\begin{framed}
  \begin{tabular}{ll}
    \textsc{Problem}: & \arpq\\
    \textsc{Input}: & A graph database $G$, and an RPQ $q = (v,regex,?x)$ over $G$.\\
    \textsc{Output}: & For each node $v'\in \semp{q}{G}$, enumerate all  shortest paths $\pi$ \\ 
     & from $v$ to $v'$ s.t. $\lambda(\pi)$ belongs to the language of $regex$. 
  \end{tabular}
\end{framed}

The algorithm for \arpq\ follows the pattern used for \srpq. Namely, we will now run Algorithm \ref{alg:plain} on $G_\times$, with $(v,q_0)$ as the starting node, where $q_0$ is the initial state of the automaton for $regex$, and returning answers only when reaching $(n,q_F)$, with $q_F$ being the accepting state of the automaton for $regex$. The algorithm for \arpq\ is presented in Algorithm \ref{alg:paths}.

In order to assure the correctness of the algorithm, we need to be careful to restrict our automaton in order not to repeat certain paths. More precisely, we will need the automaton to be unambiguous\footnote{I would like to thank Wim Martens for pointing this out to me in private communication.}, and to have a single final state. With these restrictions in mind, we get the following result.

\begin{theorem}\label{theo:corr}
Let $q = (v,regex,?x)$ be an RPQ over $G$. If the automaton $\mathcal{A}$ for $regex$ is unambiguous and has a single accepting state, then Algorithm \ref{alg:paths} will return, for each $v' \in \semp{q}{G}$, all shortest paths $\pi$ connecting $v$ and $v'$, such that $\lambda(\pi)$ is accepted by $\mathcal{A}$, \emph{precisely once}.
\end{theorem}
\begin{proof}
The fact that $\mathcal{A}$ is unambiguous assures that each path $\pi$ in $G$ can result in at most one accepting run in $\mathcal{A}$. Due to this, for each such path $\pi$ in $G$, there will be precisely one matching path in $G_\times$. This assures that no two runs match the same path. The second restriction; namely that there is only one accepting state, handles the case of automata where two paths of different length in $G$ might reach the same end node $n$, but the automaton runs they trace end in different accepting states, say $q$ and $q'$. Since the pairs $(n,q)$ and $(n,q')$ are different, the algorithm would not differentiate them on lines 14 and 18, so a longer path might be returned as one of the shortest paths reaching $n$. In case there is a single accepting state this cannot happen, since the node $n$ will already be ``tagged'' with the unique end state and the optimal distance. 
\end{proof}

One way to assure that $\mathcal{A}$ is unambiguous is to determinize it. Of course, this can potentially result in an exponential blow-up in the size of the automaton. However, in practice, the vast majority of RPQs are actually unambiguous, deterministic, or have very small deterministic version of the associated NFA \cite{BonifatiMT20}. Notice that for Algorithm \ref{alg:search}, automata need not be unambiguous, since they have a single accepting state. %, and thus the second run of the automata accepting the same path would have already been inserted into $Visited$.

The assumption of having a single accepting state   can also be lifted (both in Algorithm~\ref{alg:search} and Algorithm~\ref{alg:paths}). Namely, we could track a list of nodes that had already been reached by some (shortest) path, and in the case they get reached via a longer path through a different final state, we would discard this path. The simplest way of implementing this would be by using a dictionary which has pairs of  nodes reached thus far, together with the shortest distance to them. Then, when a node is reached via (any) final state, we would first check if it was already returned, and if so, at what distance.

Given that Algorithm \ref{alg:paths} basically runs Algorithm \ref{alg:plain} over $G_\times$, and that the enumeration is correct (Theorem \ref{theo:corr}), we obtain the following:

\begin{theorem}
Let $q = (v,regex,?x)$ be an RPQ over a graph database $G$. The problem \arpq\ can be solved in time $O(|q|\times(|V|+|E|)+|O|)$, with $O$ being the set of all outputs of the query. %In particular, after a precomputation phase running in time $O(|q|\times |G|)$, we can enumerate all answers in time proportional to their number (and length).
In other words, the problem can be solved with output linear delay and $O(|q|\times |G|)$ precomputation phase.
\end{theorem}

\subsection{Counting all shortest paths}

An interesting variation of Algorithm \ref{alg:paths} is the one where we do not reconstruct all the shortest paths, but merely count how many there are. That is, we will be solving the following problem:
\\
\begin{framed}
  \begin{tabular}{ll}
    \textsc{Problem}: & \crpq\\
    \textsc{Input}: & A graph database $G$, and an RPQ $q = (v,regex,?x)$ over $G$.\\
    \textsc{Output}: & For each node $v'\in \semp{q}{G}$, output the number of shortest paths reaching $v'$. 
  \end{tabular}
\end{framed}

  \begin{algorithm*}
  \caption{\crpq\ algorithm for a graph database $G = (V,E)$ and an RPQ $q= (v,regex,?x)$}
 \label{alg:count}
 \begin{algorithmic}[1]
 \Function{Search}{$G,q$}
 	\State $\mathcal{A}\gets Automaton(regex)$ \Comment{$q_0$ initial, $q_F$ final}
 	\State $Open.init()$ \Comment{Empty queue(BFS).}
 	\State $Visited.init()$ \Comment{Empty set of visited nodes.}
	 \State $start \gets (v,q_0,0,1)$ \Comment{A single empty path to the start node.}
% 	\If{$q_0 == q_F$} $solutions.add(n)$ \Comment{Initial state is final.}
 	%\EndIf
 	\State $Open.push(start)$
 	\State $Visited.push(start)$
 	\While{!Open.isEmpty()}
 		\State current=Open.pop() \Comment{$current = (n,q,depth,numPaths)$}
		\If{$q== q_F$} \Comment{We reached a solution}
			\State $solutions.add(n,numPaths)$ \Comment{All shortest paths already reached $n$}
		\EndIf

 		\For{next=$(n',q')$ $\in$ Neighbours(current)}
	 		\If{!(next) $\in$ Visited} \Comment{numPaths or depth are not compared for equality}
				%\State $next.pred\gets current$
				\State new = $(n',q',depth+1,numPaths)$
 				\State Open.push(new)
 				\State Visited.push(new)
 			\EndIf
 			\If{next=(n',q') $\in$ Visited} 
 			    \State new = Visited.get(n',q') \Comment{new = (n',q',depth',numPaths')} 
 			    \If{depth' == depth+1} \Comment{Another shortest path to (n',q')}
 			        \State numPaths' += numPaths \Comment{This updates the values in Visited}
 			    \EndIf
 			\EndIf
 		\EndFor 
 	\EndWhile
 \EndFunction
  \end{algorithmic}
 \end{algorithm*}
 
 This can be done by, instead of storing the list $prevList$ of pointers in Algorithm \ref{alg:paths}, simply updating the number of currently encountered shortest paths to a single node. Namely, when a new node is added to the list, all paths going from this node are added to the count. This is useful since in databases one is often interesting in having the multiplicity of the answer as well as the answer itself. Such an algorithm would then not have the memory requirement on $Visited$ as big as Algorithm \ref{alg:paths}, since instead of the list, we would simply store a single number indicating how many shortest paths there are reaching a certain node from $(v,q_0)$. Again, for unambiguous automata, this would then amount to shortest paths in the original graph database $G$ as well. We present this version of the algorithm in Algorithm~\ref{alg:count}. The states stored in $Open$ and $Visited$ are now of the form $(n,q,depth,numPaths)$, with $n,q$ and $depth$ as in Algorithm~\ref{alg:paths}, and $numPaths$ counting the number of paths encountered thus far. It straightforward to extend all of the presented algorithms to also cover the case of 2RPQs~\cite{CalvaneseGLV02}, which allow the edges to be traversed backwards.

\section{Conclusions and looking ahead}
\label{sec:concl}
In this paper we provided an illustration of how the classical BFS algorithm can be extended in order to also reconstruct all shortest paths from a single source node to all the reachable nodes. We then applied this algorithm in the setting of graph databases in order to evaluate regular path queries, extending them with the option to return all shortest paths, instead of a single path.
In the future, we would like to explore whether this algorithm can be extended in order to also reason on simple paths, trails, or to return all paths up to a certain length, ordered from shortest to longest. We believe such contribution can be important in order to better understand the landscape of different query modes for regular path queries. 

In terms of implementation, all of the presented algorithms have already been implemented within MillenniumDB~\cite{MillDB}, a recent open-source graph database engine. An interested reader can try these out by accessing the MillenniumDB code at \url{https://github.com/MillenniumDB/MillenniumDB}.
%
%We will also explore how to implement these algorithm in the context of an actual graph database engine, and are already doing so within MillenniumDB~\cite{MillDB}, an open-source graph database with traditional relation architecture.

\noindent{\textbf{Acknowledgements.}}
This work was supported by ANID -- Millennium Science Initiative Program -- Code ICN17\_002. I would also like to thank Carlos Rojas, Juan Romero and Wim Martens for their input.


\bibliography{nourlbiblio}

%\bibliographystyle{ACM-Reference-Format}

%UNCOMMENT IF FULL BIB IS DESIRED
% \bibliography{biblio}
%\bibliography{nourlbiblio}


\end{document}